\documentclass[conference]{IEEEtran}
\IEEEoverridecommandlockouts

\usepackage{amsmath,amssymb,amsthm,mathtools}
\usepackage{bm}            
\usepackage{graphicx}
\usepackage{booktabs}
\usepackage{array}
\usepackage{comment}
\usepackage[hidelinks]{hyperref}
\usepackage{amsthm}
\usepackage{cite}
\usepackage{color}
\usepackage{xcolor}

\newtheorem{theorem}{Theorem}

\usepackage{algorithm}
\usepackage{algpseudocode}   
\usepackage{bm}
\usepackage{subcaption}

\newcommand{\hh}{\bm{h}}

\makeatletter
\algrenewcommand\alglinenumber[1]{\footnotesize #1}
\makeatother

\algrenewcommand\algorithmicrequire{\textbf{Input:}}
\algrenewcommand\algorithmicensure{\textbf{Output:}}

\allowdisplaybreaks[1]
\newcommand{\E}{\mathbb{E}}
\newcommand{\norm}[1]{\left\lVert #1 \right\rVert}

\newcommand{\tth}{\bm{\theta}}
\newcommand{\tths}{\bm{\theta}^\star}

\newcommand{\tthtp}{\bm{\theta}(t{+}1)}

\newcommand{\vvtp}{\bm{\nu}(t{+}1)}
\newcommand{\Dth}{\Delta\bm{\theta}}

\newcommand{\Dthhat}{\widehat{\Delta\bm{\theta}}}
\newcommand{\ecomp}{\bm{e}_{\mathrm{comp}}}
\newcommand{\ech}{\bm{e}_{\mathrm{channel}}}

\newcommand{\stepsize}{\eta}
\newcommand{\local}{\tau}
\newcommand{\M}{M}
\newcommand{\N}{N}

\begin{document}

\title{Age-Aware Edge-Blind Federated Learning via Over-the-Air Aggregation}

\author{$\hspace{-.25in}$Ahmed M. Elshazly$^{1,2}\qquad$Ahmed Arafa$^1$\\$^1$Department of Electrical and Computer Engineering, University of North Carolina at Charlotte, NC 28223\\
$^2$Electrical Engineering Department, Faculty of Engineering, Alexandria University, Alexandria, Egypt\\
\emph{aelshaz1@charlotte.edu}$\qquad$\emph{aarafa@charlotte.edu}
\thanks{This work was supported by the U.S. National Science Foundation under Grants CNS 21-14537 and ECCS 21-46099.}}

\maketitle

\begin{abstract}
We study federated learning (FL) over wireless fading channels where multiple devices simultaneously send their model updates. We propose an efficient \emph{age-aware edge-blind over-the-air FL} approach that does not require channel state information (CSI) at the devices. Instead, the parameter server (PS) uses multiple antennas and applies maximum-ratio combining (MRC) based on its estimated sum of the channel gains to detect the parameter updates. A key challenge is that the number of orthogonal subcarriers is limited; thus, transmitting many parameters requires multiple Orthogonal Frequency Division Multiplexing (OFDM) symbols, which increases latency. To address this, the PS selects only a small subset of model coordinates each round using \emph{AgeTop-\(k\)}, which first picks the largest-magnitude entries and then chooses the \(k\) coordinates with the longest waiting times since they were last selected. This ensures that all selected parameters fit into a single OFDM symbol, reducing latency. We provide a convergence bound that highlights the advantages of using a higher number of antenna array elements and demonstrates a key trade-off: increasing \(k\) decreases compression error at the cost of increasing the effect of channel noise. Experimental results show that (i) more PS antennas greatly improve accuracy and convergence speed; (ii) AgeTop-\(k\) outperforms random selection under relatively good channel conditions; and (iii) the optimum \(k\) depends on the channel, with smaller \(k\) being better in noisy settings.
\end{abstract}

\section{Introduction}

Federated Learning (FL) \cite{H. B. McMahan 2016} is a distributed training approach to machine learning in which the training process is done at the edge devices and only the local updates are transmitted to a parameter server (PS) that orchestrates the learning process. Hence, FL solves multiple limitations of centralized (classical) learning with respect to efficient transmission, storage and privacy \cite{Federated learning: Strategies, privacy, advances and problems fl}.

FL was first implemented in a sequential manner over digital links, e.g., \cite{H. B. McMahan 2016}. Although this approach is reliable, it has efficiency limitations, especially when the number of devices and model parameters are large \cite{analog-digital, delay sensitive}. To improve efficiency, later works proposed selecting only a subset of devices per round \cite{ABS, value of information, submodular,online scheduling, two-phase deep, CSMAAFL} or transmitting quantized updates to reduce communication overhead \cite{QSGD, UVeQFed, quantization constraints, joint privacy quantization, Hierarchical quantization, DSGD, FedAQ}. Recently, analog over-the-air (OTA) aggregation has been explored, exploiting the wireless multiple access channel (MAC)’s superposition property to allow simultaneous update transmissions that directly aggregate at the PS~\cite{OTA computation, Beamforming and device selection OTA, Beamforming Vector Design, Active RIS}.  

For accurate aggregation of the OTA local updates, each device must perfectly estimate its instantaneous channel state information (CSI) and apply phase correction and channel inversion before transmission \cite{Knowledge-guided OTA, Age-based device selection, on analog, Broadband analog}. As a result, the local updates are received with almost equal power and constructively aggregate the PS. However, channel estimation is computationally intensive, increasing communication overhead, power consumption, and delay. To overcome these challenges, several approaches have been proposed to eliminate channel estimation at the device/client side \cite{NCAirFL,Revisiting analog,Random Orthogonalization, Distributed SGD}. Although the OTA approach is efficient, the transmission latency increases significantly with the number of model parameters. This has been mitigated in the literature by multiple approaches, such as compressive sensing (CS) methods \cite{CS 1, CS 2}, random parameter selection \cite{rtopk}, and intelligent index-selection strategies, including layer-wise sparsification, gradient segmentation, and age-based selection \cite{Layer, SegOTA}.

The most closely-related works to ours are \cite{blind} and \cite{age aware}. In \cite{blind}, the authors propose a blind FL approach using a multiple-antenna PS, in which the entire model updates are transmitted using multiple OFDM symbols. In their approach, the potential advantages of model parameter selection to reduce latency and noise effects has not been investigated. On the other hand, the authors in \cite{age aware} propose an age-aware OTA algorithm for a single-antenna PS, in which the client gradients are compressed to a smaller dimension and sent via orthogonal waveforms over a simplified flat-fading channel with no power constraints. While this represents a simpler channel model, their main objective was to highlight the advantage of model parameter selection using \emph{AgeTop-\(k\)}, which first picks the indices corresponding to the largest magnitude entries and then chooses from them the $k$ coordinates with the longest waiting times since they were last selected.

In this paper, we propose an efficient edge-blind OTA FL approach in which no CSI estimation is required at the client devices. Instead, the PS is equipped with multiple antennas and applies maximum ratio combining (MRC) based on its estimated sum of the channel gains. Our framework is integrated with the AgeTop-$k$ index selection approach to improve communication efficiency by transmitting a single OFDM symbol over a realistic frequency selective channel with power constraints. Moreover, we investigate the impact of the number of parameters selected for transmission on the system performance. Our convergence analysis and experimental results demonstrate a trade-off between transmission quality and compression effectiveness with respect to the number of selected parameters $k$: larger values of $k$ minimize compression errors, yet maximize the effect of noise. We show there exists an optimal $k$ that strikes a balance in between.

\section{System Model}

We describe the system model by first presenting the FL process, followed by the communication model and the compression technique used for efficient transmission.

\subsection{Federated Learning Model}

The FL system consists of $M$ client devices and a PS that aggregates the model updates over a wireless channel. The objective is to optimize the model parameters $\bm{\theta} \in \mathbb{R}^d$ that minimize the global loss $F(\bm{\theta})$. That is, to solve for
\begin{equation}
\bm{\theta}^\star \triangleq \arg\min_{\bm{\theta}} F(\bm{\theta}).
\label{eq:theta_star}
\end{equation}
The global loss is computed as the weighted sum of each client’s average empirical loss, with weights given by the ratio of the client’s dataset size to the total dataset size:
\begin{equation}
F(\bm{\theta}) = \sum_{m=1}^M \frac{B_m}{B} F_m(\bm{\theta}). \label{eq:F_theta}
\end{equation}
Here, the local dataset of device $m$ is denoted by $\mathcal{B}_m$ with $B_m \triangleq |\mathcal{B}_m|$, and the total dataset size is $B \triangleq \sum_{m=1}^M B_m$. The average empirical loss at device $m$ is
\begin{equation}
F_m(\bm{\theta}) = \frac{1}{B_m} \sum_{\bm{b} \in \mathcal{B}_m} f(\bm{\theta}, \bm{b}), \quad m \in [M], \label{eq:Fm_theta}
\end{equation}
where $f(\bm{\theta}, \bm{b})$ denotes the loss associated with sample $\bm{b}$ for model parameters $\bm{\theta}$. During global iteration $t$, after receiving $\bm{\theta}(t)$ from the PS, device $m$ performs $\tau$ local stochastic gradient descent (SGD) steps:
\begin{equation}
\bm{\theta}^{i+1}_m(t) = \bm{\theta}^i_m(t) - \eta(t)\, \nabla F_m\!\left( \bm{\theta}^i_m(t), \xi^i_m(t) \right), \quad i \in [\tau], \label{eq:local_update}
\end{equation}
where $\eta(t)$ represents the learning rate and $\xi^i_m(t)$ denotes the mini-batch. Each device sends $\Delta\bm{\theta}_m(t) = \bm{\theta}^{\tau+1}_m(t) - \bm{\theta}(t)$ after completing its local updates. Ideally, when accurate local updates are available, the PS updates the global model as
\begin{equation}
\bm{\theta}(t+1) = \bm{\theta}(t) + \Delta\bm{\theta}(t),
\end{equation}
where the aggregated update $\Delta\bm{\theta}(t)$ is the average given by
\begin{equation}
\Delta\bm{\theta}(t) \triangleq \frac{1}{M} \sum_{m=1}^M \Delta\bm{\theta}_m(t). \label{eq:global_update}
\end{equation}

\subsection{Age-Aware Over-the-Air Communication Model}

Practically, the local model updates of the clients are transmitted to the PS through a wireless channel and are aggregated over-the-air (OTA). Each client device is equipped with a single antenna, while the PS is equipped with an $N$-element antenna array for reception. The channel is modeled as a frequency selective fading channel, and OFDM is employed to transmit the local updates of each client over $s$ orthogonal subcarriers. We consider the case in which the number of orthogonal subcarriers is limited \( (s \ll d) \), thereby transmitting the entire local update of each device would incur high latency. Therefore, a compression step is implemented in which the local update vector of each client is compressed to dimension $k$, with \( k \ll d \), using an {\it age-aware} selection algorithm, under the assumption that  $s=k/2$.

Specifically, in round $t$, the PS maintains two state variables:  
$(i)$ a global update buffer $\Dth_{\mathrm{global}}(t)\in\mathbb{R}^{d}$ that holds the most recent estimate of each coordinate, and  
$(ii)$ an age-of-information (AoI) vector $\bm{a}(t) \in \mathbb{N}^{d}$ that tracks the staleness of each coordinate, defined as the number of communication rounds since it was last updated. Using $\Dth_{\mathrm{global}}(t)$, the server first identifies the indices of the $r$ largest-magnitude entries:
\begin{equation}
\mathcal{S}_r(t)\ \triangleq\ 
\operatorname{arg\,Top}\text{-}r\Big(\,\big|\Dth_{\mathrm{global}}(t)\big|\,\Big).
\label{eq:topr}
\end{equation}
Within this candidate set, the server then applies AgeTop-$k$ selection by choosing the $k$ indices with the largest ages:
\begin{equation}
\mathcal{S}_k(t)\ \triangleq\ 
\operatorname{arg\,Top}\text{-}k\Big(\,\{\,a_j(t) : j\in \mathcal{S}_r(t)\,\}\Big).
\label{eq:topk_aoi}
\end{equation}
Thus, coordinates are selected if they are both significant in magnitude and sufficiently stale.  

Let $\bm{S}(t) \in \{0,1\}^{d \times d}$, which is broadcast together with the current model $\tth(t)$, be a diagonal matrix such that $(\bm{S}(t))_{j,j} = 1$ if $j \in \mathcal{S}_k(t)$. Removing the $d-k$ all-zero rows yields the reduced selector $\hat{\bm{S}}(t) \in \{0,1\}^{k \times d}$. Each client $m$ then computes its local update $\Delta \tth_m(t)$ and compresses it to

\begin{equation}
\bm{u}_m(t) \triangleq \hat{\bm{S}}(t)\,\Dth_m(t)\in\mathbb{R}^{k}.
\end{equation}

For an even $k$ (if $k$ is odd, it is incremented by $1$ to make it even), 
we define $s = k/2$ and split $\bm{u}_m(t)$ into 
$\bm{u}_{m,\mathrm{re}}(t), \bm{u}_{m,\mathrm{im}}(t) \in \mathbb{R}^{s}$. These are modulated into a single OFDM symbol with $s$ orthogonal subcarriers as
\begin{equation}
\bm{x}_m(t) \triangleq 
\alpha_t\big(\bm{u}_{m,\mathrm{re}}(t)+\mathrm{j}\,\bm{u}_{m,\mathrm{im}}(t)\big)\in\mathbb{C}^{s},
\label{eq:xm_def}
\end{equation}
with $\alpha_t$ chosen to satisfy the per-round transmit power constraint $\bar{P}$:
\begin{equation}
    \bigl\lVert \alpha_t \, \bm{u}_m(t) \bigr\rVert_2^{2} \;\le\; \bar{P}, 
    \quad \forall \, m, t .
\label{eq:per_round_power}
\end{equation}
At antenna $n\in[\N]$, the received vector is
\begin{equation}
\bm{y}_{n}(t) 
= \sum_{m=1}^{\M}\hh_{m,n}(t)\odot \bm{x}_m(t)+\bm{z}_{n}(t),
\label{eq:ant_rx}
\end{equation}
where $\hh_{m,n}(t)\!\in\!\mathbb{C}^{s}$ is the channel gain vector, with each element $h_{m,n,i}(t)$ distributed as $\mathcal{CN}(0,\sigma_h^2)$ for $i \in [s]$, and $\bm{z}_{n}(t)\!\in\!\mathbb{C}^{s}$ is additive white Gaussian noise, with each element $z_{n,i}(t)$ distributed as $\mathcal{CN}(0,\sigma_z^2)$ for $i \in [s]$. 

The PS is assumed to have perfect CSI of the sum of the channel,
i.e., exact knowledge of $\left(\sum_{m=1}^{M} \hh_{m,n}(t)\right)$.\footnote{This has been a reasonable assumption in the literature, e.g., \cite{Random Orthogonalization, blind}.} 
With this, the normalized MRC output is
\begin{equation}
\bm{y}(t) 
= \frac{1}{\alpha_t\,\M\,\sigma_h^{2}\,\N}
\sum_{n=1}^{\N}\Big(\sum_{m=1}^{\M}\hh_{m,n}(t)\Big)^{\!*}\odot \bm{y}_{n}(t).
\label{eq:mrc}
\end{equation}
As the number of antennas $N$ increases, the performance of the MRC receiver improves as interference and noise terms diminish with $N$, and the best performance is achieved when $N \to \infty$~\cite{blind}. In practice, a sufficiently large but finite number of antennas yields an accurate approximation of the true average of the local updates
\begin{equation}
\bm{y}(t)\ \approx\ \frac{1}{\M}\sum_{m=1}^{\M}
\big(\bm{u}_{m,\mathrm{re}}(t)+\mathrm{j}\,\bm{u}_{m,\mathrm{im}}(t)\big).
\label{eq:y_avg}
\end{equation}
The PS reconstructs the $k$ compressed entries as
\begin{equation}
\widehat{\bm{u}}(t) = 
\big[\Re\{\bm{y}(t)\};\,\Im\{\bm{y}(t)\}\big]\in\mathbb{R}^{k},
\label{eq:u_hat}
\end{equation}
and decompresses that to the $d$-dimensional space:
\begin{equation}
\widehat{\Dth}(t)=(\hat{\bm{S}}(t))^{\!\top}\widehat{\bm{u}}(t).
\label{eq:theta_hat}
\end{equation}
The reconstructed and decompressed update  $\widehat{\Dth}(t)$ is then used to update the global model as
\begin{equation}
 \tth(t{+}1) = \tth(t) + \widehat{\Dth}(t).
\label{eq:update rule}
\end{equation}
Finally, the global update buffer is refreshed as
\begin{equation}
\Dth_{\mathrm{global},j}(t{+}1) \triangleq
\begin{cases}
\widehat{\Dth}_j(t), & j\in\mathcal{S}_k(t),\\[0.3em]
\Dth_{\mathrm{global},j}(t), & j\notin\mathcal{S}_k(t),
\end{cases}
\label{eq:delta_global_update}
\end{equation}
while the AoI counters evolve as
\begin{equation}
a_j(t{+}1) \triangleq
\begin{cases}
0, & j\in\mathcal{S}_k(t),\\[0.3em]
a_j(t)+1, & j\notin\mathcal{S}_k(t).
\end{cases}
\label{eq:age_update}
\end{equation}
Hence, selected coordinates are refreshed with the newly aggregated updates, while non-selected coordinates keep their previous values and accumulate age. This approach maintains the balance between magnitude importance and staleness across iterations. We denote our proposed approach as {\it age-aware edge-blind OTA FL}, and summarize it in Algorithm~\ref{alg:ageblind}.

\begin{algorithm}[]
\caption{Age-Aware Edge-Blind OTA FL}
\label{alg:ageblind}
\small
\begin{algorithmic}[1]  
\Require Initial model $\tth(0)$; rounds $T$; Top-$r$ (magnitude); Top-$k$ (AoI); antennas $\N$
\State Initialize AoI $a[1{:}d]\gets \mathbf{0}$ 
\For{$t=0$ to $T-1$}
    \State $\mathcal{S}_r(t) \gets \operatorname{arg\,Top}\text{-}r\!\left(|\Dth_{\mathrm{global}}(t)|\right)$
    \State $\mathcal{S}_k(t) \gets \operatorname{arg\,Top}\text{-}k\text{-AoI}\!\left(\mathcal{S}_r(t);\, a(t)\right)$
    \State Construct diagonal $\bm{S}(t)$ s.t.\ $S_{j,j}=1$ if $j\!\in\!\mathcal{S}_k(t)$
    \State Remove zero rows of $\bm{S}(t)$ to get reduced selector $\hat{\bm{S}}(t)$
    \State Broadcast $(\tth(t),\, \hat{\bm{S}}(t))$ to all clients
    \ForAll{$m \in \{1,\dots,\M\}$}
        \State $\tau$ local SGD steps using~(\ref{eq:local_update})
        \State$\Delta\bm{\theta}_m(t) = \bm{\theta}^{\tau+1}_m(t) - \bm{\theta}(t)$
        \State Compute and transmit $\bm{x}_m(t)$ using~(\ref{eq:xm_def})
    \EndFor
    \State Apply MRC using~(\ref{eq:ant_rx}) and~(\ref{eq:mrc}) 
    \State Reconstruct and decompress the updates using~(\ref{eq:u_hat}) and~(\ref{eq:theta_hat})
    \State Update the global model using~(\ref{eq:update rule})
    \State Update AoI using (19)    
    \State Refresh the  global buffer $\Dth_{\mathrm{global}}(t)$ using~(\ref{eq:delta_global_update})
\EndFor
\end{algorithmic}
\normalsize
\end{algorithm}

\section{Convergence Analysis}

In this section, we provide a convergence analysis of the proposed age-aware edge-blind OTA FL technique.

\subsection{Preliminaries}

We make the following assumptions and definitions.

\noindent\textbf{Assumption 1:} 
Loss function $F_m$ is $L$-smooth $\forall\, m \in [M]$; i.e.,
$\forall\, \bm{v}, \bm{w} \in \mathbb{R}^d$,
\begin{equation}
    F_m(\bm{v}) - F_m(\bm{w}) 
    \le \langle \bm{v} - \bm{w}, \nabla F_m(\bm{w}) \rangle 
    + \frac{L}{2} \| \bm{v} - \bm{w} \|_2^2.
\end{equation}

\noindent\textbf{Assumption 2:} 
Loss function $F_m$ is $\mu$-strongly convex $\forall\, m \in [M]$; i.e., 
$\forall\, \bm{v}, \bm{w} \in \mathbb{R}^d$,
\begin{equation}
    F_m(\bm{v}) - F_m(\bm{w}) 
    \ge \langle \bm{v} - \bm{w}, \nabla F_m(\bm{w}) \rangle 
    + \frac{\mu}{2} \| \bm{v} - \bm{w} \|_2^2.
\end{equation}

\noindent\textbf{Assumption 3:} 
The mean squared $\ell_2$-norms of the gradients are bounded above by $G^2$; i.e., 
$\forall\, i \in [\tau]$, $\forall\, m \in [M]$, $\forall\, t$,
\begin{equation}
    \mathbb{E} 
    \left[ \left\| \nabla F_m\!\left( \bm{\theta}_m^{i}(t) \right) \right\|_2^2 \right] 
    \le G^2.
\end{equation}

\noindent\textbf{Assumption 4:} 
The variance of the stochastic gradient is uniformly bounded above by $\sigma^2$; i.e., 
for all $m \in [M]$, $\bm{\theta} \in \mathbb{R}^d$, and data samples $\xi$,
\begin{equation}
    \mathbb{E}\!\left[
        \left\| 
            \nabla F_m\!\big(\tth_m^{\,i}(t),\, \xi_m^{i}(t)\big)
            - \nabla F_m\!\left( \bm{\theta}_m^{i}(t) \right) 
        \right\|_2^2
    \right] 
    \le \sigma^2 .
\end{equation}

\noindent\textbf{Assumption 5:}  
For $\big|\Dth_{\mathrm{global}}(t)\big|$, the ratio of the largest magnitude value to the $r$-th largest magnitude value is bounded above by $\beta$ for all iterations $t$.
\vspace{0.5em}

\noindent\textbf{Definition:}  
We define the {\it bias in data distribution} as 
\begin{equation}
    \Gamma \triangleq F^{\ast} - \sum_{m=1}^{M} \frac{B_{m}}{B} \, F_{m}^{\ast},
\end{equation}
where $F^{\ast}$ and $F_{m}^{\ast}$ represent the minimum values of the global loss function and the $m$-th local loss function, respectively.

\subsection{Convergence Rate}

We now present the convergence upper bound of the proposed approach. Following~\cite{ragek}, the algorithm termed \emph{rAge-$k$} can be viewed as a compression operator $\mathrm{comp}_{k}(\cdot)$ with
\begin{equation}
    \gamma = \frac{k}{k + (r - k)\beta + (d - r)},
    \label{eq:gamma_def}
\end{equation}
where $\gamma$ is a constant satisfying
\begin{equation}
    \mathbb{E} \left[ \left\| \bm{\theta} - \mathrm{comp}_{k}(\bm{\theta}) \right\|_2^2 \right]
    \le (1 - \gamma) \, \left\| \bm{\theta} \right\|_2^2.
    \label{eq:compression_def}
\end{equation}
The next theorem specifies the convergence bound. 

\begin{theorem}\label{thm:main}
Let $0 < \eta(t) \le \min\!\left\{ 1, \frac{1}{\mu \tau} \right\}, \; \forall t$. Under Assumptions~1--5, we have
\begin{align}
    \mathbb{E} \left[ F\!\left( \bm{\theta}(T) \right) \right] - F^{\ast}
    &\le \frac{L}{2} \, \mathbb{E} \left[ \left\| \bm{\theta}(T) - \bm{\theta}^{\ast} \right\|_2^2 \right] \\
    &\le \frac{L}{2} 
    \left( \prod_{i=0}^{T-1} D(i) \right) 
    \left\| \bm{\theta}(0) - \bm{\theta}^{\ast} \right\|_2^2 \notag \\
    &\quad + \frac{L}{2} 
    \sum_{j=0}^{T-1} Q(j) 
    \prod_{i=j+1}^{T-1} D(i).
\end{align}
The terms \( D(i) \) and \( Q(i) \) are defined as follows: \begin{IEEEeqnarray}{rCl}
D(i) &=& 2\,A(i) 
 \;=\; 2\Big[1-\mu\,\eta(i)\big(\tau-\eta(i)\,(\tau-1)\big)\Big], 
 \label{eq:D_def}\\[1ex]
Q(i) &=& 2(1-\gamma)\,\eta^{2}(i)\tau^{2}\big(G^{2}+\sigma^{2}\big)
  + \frac{\eta^{2}(i)\tau^{2}G^{2}}{N} \nonumber\\
&& {}+ 2\big(1+\mu(1-\eta(i))\big)\eta^{2}(i)G^{2}
      \frac{\tau(\tau-1)(2\tau-1)}{6} \nonumber\\
&& {}+ 2\,\eta^{2}(i)\!\left( (\tau^{2}+\tau-1)G^{2}
                           + 4\,\eta(i)(\tau-1)\Gamma \right) \nonumber\\
&& {}+ \frac{\sigma_{z}^{2}k}{2\,\alpha_{i}^{2}\,N\,M\,\sigma_{h}^{2}}.
\label{eq:Q_def}
\end{IEEEeqnarray}
\end{theorem}

\begin{proof}[Proof Sketch]

We first define $\vvtp \triangleq \tth(t)+\Dth(t)$. Using \eqref{eq:update rule}, we write
\begin{align}
\|\tthtp-&\tths\|^{2}
= \underbrace{\|\tthtp-\vvtp\|^{2}}_{A}
 + \underbrace{\|\vvtp-\tths\|^{2}}_{B} \nonumber \\
& + \underbrace{2\,\big\langle \tthtp-\vvtp,\, \vvtp-\tths \big\rangle}_{C}.
\end{align}
Towards bounding the terms $A$, $B$ and $C$, we express 
\begin{align}
\Dthhat(t)=\Dth(t)+\ecomp(t)+\ech(t),
\end{align}
where $\ecomp(t)=\mathrm{comp}_k(\Dth(t))-\Dth(t)$ denotes the compression error, and $\ech(t)$ denotes the channel error.

Now let us consider the stochastic gradient $\nabla F_m\!\left(\tth_m^{\,i}(t),\,\xi_m^{i}(t)\right)$ at client $m$. We assume this is an unbiased estimate of the true gradient. Using Assumptions~3 and~4, one can show that

\begin{equation}
\E\!\left[
    \left\|
        \nabla F_m\!\left(\tth_m^{\,i}(t),\,\xi_m^{i}(t)\right)
    \right\|^2
\right]
\;\le\; G^2 + \sigma^2.
\label{eq:gmibound}
\end{equation}
Using this, one can apply some algebraic manipulations to further write
\begin{IEEEeqnarray}{rCl}
\E\!\left[\norm{\Dth(t)}^{2}\right]
&\le& \stepsize^{2}\,\local^{2}\,(G^{2}+\sigma^{2}). 
\label{eq:avg-delta-bound}
\end{IEEEeqnarray}
Thus, one can bound the compression error as
\begin{align}
\E\!\left[\norm{\ecomp(t)}^{2}\right]
&\le (1-\gamma)\,\stepsize^{2}\,\local^{2}\,(G^{2}+\sigma^{2}).
\label{eq:comp-bound}    
\end{align}
By \cite[Lemma~1]{blind}, one can show that
\begin{align} \label{eq:channel-error}
\E\norm{\ech(t)}^{2}
&\le \frac{\stepsize^{2}(t)\local^{2}G^{2}}{\N}
   +\frac{\sigma_{z}^{2}k}{2\,\alpha_{t}^{2}\,N\,M\,\sigma_{h}^{2}}.
\end{align}

Next, we show that the term $A$ satisfies the following:
\begin{equation}
\E[A] \;=\; \E\norm{\ecomp(t)}^{2}+\E\norm{\ech(t)}^{2},
\label{eq:EA-sum}
\end{equation}
which is readily bounded using \eqref{eq:comp-bound} and \eqref{eq:channel-error}. Similarly, using some extensive manipulations, and by \cite[Lemma~2]{blind}, we show that the term $C$ satisfies
\begin{equation}
\E[C]\;\le\; \E\norm{\ecomp(t)}^{2}
+ A(t)\,\E\norm{\tth(t)-\tths}^{2}
+ B(t),
\end{equation}
where
\begin{IEEEeqnarray}{rCl}
A(t)&=&\Big[1-\mu\,\stepsize(t)\big(\local-\eta(t)(\local-1)\big)\Big],\\
B(t)&=&\big(1+\mu(1-\stepsize(t))\big)\stepsize^{2}(t)G^{2}\frac{\local(\local-1)(2\local-1)}{6} \nonumber\\
&& {}+\stepsize^{2}(t)\!\left((\local^{2}+\local-1)G^{2}+2\,\stepsize(t)(\local-1)\Gamma\right).
\end{IEEEeqnarray}
Finally, the term $B$ is shown to satisfy
\begin{align}
    \E[B]\le A(t)\,\E\norm{\tth(t)-\tths}^{2}+B(t).
\end{align}
Combining the bounds, expanding and rearranging using the definitions of $D(i)$ and $Q(i)$ in the theorem provides the result.
\end{proof}

\noindent\textbf{Remark 1:}\label{remark:2}
Compared with the convergence upper bound in~\cite{blind}, the bound in Theorem~\ref{thm:main} contains the additional term
\(
2(1-\gamma)\,\eta^{2}(i)\,\tau^{2}\bigl(G^{2}+\sigma^{2}\bigr),
\)
which arises from the compression error. This term decreases as \(k\) increases, as expected, since $\gamma$ increases with $k$. In contrast, the channel noise contribution term 
\(
{\sigma_{z}^{2}\,k}/{2\,\alpha_{i}^{2}\,N\,M\,\sigma_{h}^{2}}
\)
increases (linearly) in \(k\). Hence, there exists an optimal value of the dimension of the compressed model update $k$ that balances compression errors and channel noise effects. We elaborate on that in the next section.

\begin{figure}[t]
    \centering
    \includegraphics[width=0.5\textwidth]{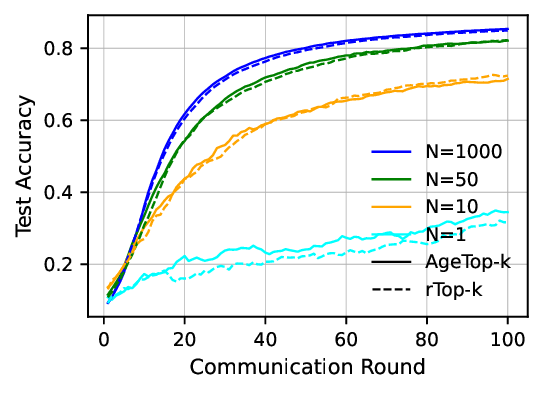}
    \caption{Performance with different number of antennas $N$.}
    \label{fig:N}
\end{figure}

\section{Experimental Results and Discussion}

In this section, we conduct several experiments to further analyze the proposed age-aware edge-blind OTA FL technique. In our experiments, we use the MNIST dataset and perform training using logistic regression implemented as a fully connected single-layer neural network consisting of 10 nodes, resulting in a total of 7,850 parameters. The SGD optimizer is utilized with a fixed learning rate $\eta=0.001$ and a batch size of 32. The global model is updated over $T=100$ communication rounds. To satisfy the power constraint, the power scaling parameter $\alpha_t$ is chosen based on the client with the maximum update norm as follows:
\begin{equation}
    \alpha_t \;=\; 
    \sqrt{ \frac{\bar{P}}{\displaystyle \max_{m} \bigl\lVert \bm{u}_m(t) \bigr\rVert_2^{2}} } \, .
\label{eq:alpha_definition}
\end{equation}
This can be achieved in practice, e.g., by having all clients send their update norms over a control channel, and the PS selecting the value of $\alpha_t$ based on the maximum among them. All experiments are conducted $5$ times and the average performance is reported. We compare our algorithm's performance against \texttt{rTop-k}~\cite{rtopk}, in which after identifying the $r$ largest parameter values in magnitude, the PS randomly selects $k$ of them to request from clients. In the figures, we denote our algorithm as \texttt{AgeTop-k}.

We first investigate the effect of the number of antennas $N$. We set $M=10$ clients and consider an extreme non-IID scenario in which   each client is given training samples of only one MNIST data label. We set $\tau=3$ local SGD iterations to train and update the model. For compression, with a slight abuse of notation, the value of $r$ is set to $0.9$ of the total number of parameters (i.e., $r$ here represents a ratio); likewise, we set the value of $k$ to $0.2$ of the total number of parameters. The wireless channel parameters are set to $\sigma_h^2=1$ and $\sigma_z^2=5$, while each client device is equipped with a power budget of $\bar{P}=10$. The results in Fig.~\ref{fig:N} show that using only a single antenna is insufficient for the system to converge to an acceptable test accuracy. This is because the interference and noise effects of the received signals cannot be mitigated with only one antenna, leading to a significant degradation in training performance. In contrast, increasing the number of antenna elements significantly improves convergence speed and test accuracy, since having more antennas makes the reconstructed received signal at the PS closer to the true average of the client updates. The results in Fig. \ref{fig:N} also show that \texttt{AgeTop-k} achieves better convergence rate and accuracy than \texttt{rTop-k} by using relatively larger number of antennas ($N=50$ and $N=1000$), since this helps mitigate the randomness introduced by the channel, and represents relatively good channel conditions. However, when the number of the receiving antenna elements is relatively small ($N=10$), the two techniques achieve a similar performance, since the randomness introduced by the channel cannot be mitigated using only $10$ antenna elements. 

\begin{figure}[t]
    \centering
    \includegraphics[width=0.5\textwidth]{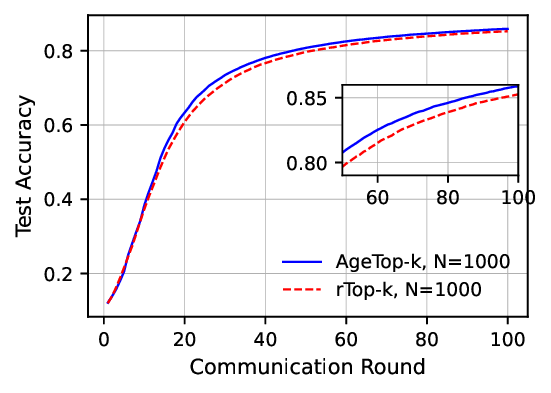}
    \caption{\texttt{AgeTop-k} versus \texttt{rTop-k} under relatively good channel conditions.}
    \label{fig:age_random}
    \vspace{-.2in}
\end{figure}

We further emphasize the advantage of \texttt{AgeTop-k} over the baseline in another experiment conducted by ensuring that the global update vector is constructed after facing {\it good} channel conditions to limit the randomness resulting from the channel. Here, the value of $\bar{P}$ is increased to $10$ and the value of $\sigma_z^2$ is decreased to $0.1$. We set $N=1000$ antennas. The results in Fig.~\ref{fig:age_random} demonstrate that the smarter way of selecting the transmitted indices under \texttt{AgeTop-k} boosts the convergence rate and the test accuracy of the system. 

\begin{figure}[t]
    \centering
    \begin{subfigure}{0.45\textwidth}
        \centering
        \includegraphics[width=\linewidth]{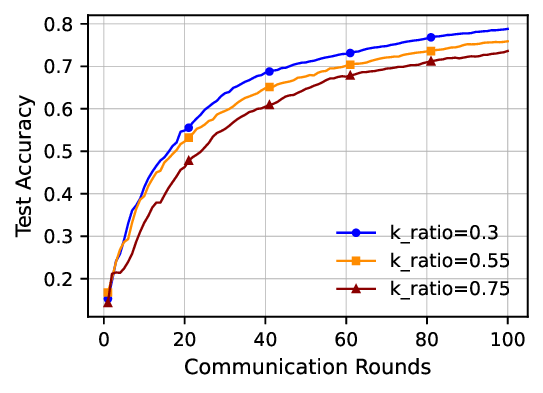}
        \caption{Severe channel conditions.}
        \label{fig:k_a}
    \end{subfigure}
    \hfill
    \begin{subfigure}{0.45\textwidth}
        \centering
        \includegraphics[width=\linewidth]{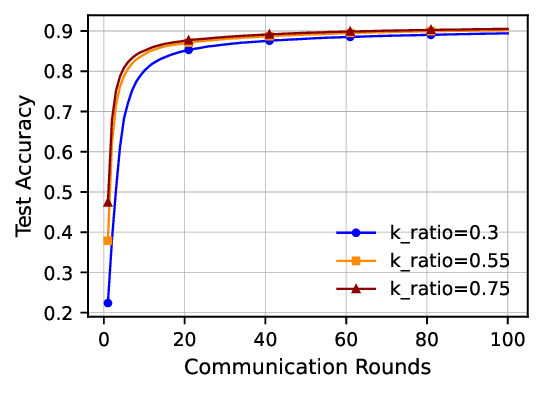}
        \caption{Good channel conditions.}
        \label{fig:k_b}
    \end{subfigure}

    \caption{Performance with different compression values $k$.}
    \label{fig:k}
    \vspace{-.2in}
\end{figure}

Next, we investigate the effect of choosing the best $k$ to balance the performance as indicated in Remark~\ref{remark:2}. In this experiment, the training dataset is distributed across $M=20$ clients in an IID manner. We set $\tau=1$ SGD update, $N=10$ and $r=0.75$. The results in Fig.~\ref{fig:k_a} are obtained using $\bar{P}=1$, $\sigma_h^2=1$ and $\sigma_z^2=50$. Fig. \ref{fig:k_a} shows that using a relatively low value of $k=0.3$ achieves the best performance. This is mainly due to the channel conditions being severe: the noise variance is very high compared to the transmit power budget per device. As indicated in Remark~\ref{remark:2}, in this case the noise term of the convergence upper bound  dominates the convergence rate and hence decreasing the value of $k$ leads to better performance. In Fig. \ref{fig:k_b} we change the channel conditions and set $\bar{P}=10$, $\sigma_h^2=1$ and $\sigma_z^2=1$. Now we see that a relatively high value of $k=0.75$ performs best. This is consistent with Remark~\ref{remark:2}'s  indication that under good channel conditions the compression term in the convergence upper bound dominates performance, and hence increasing the value of $k$ is best.

\section{Conclusion}

This paper introduced an age-aware edge-blind over-the-air FL framework that is efficient and avoids the need for client-side CSI estimation. The PS aggregates the transmitted updates using multi-antenna MRC and updates only a small subset of coordinates through a two-step technique (AgeTop-$k$): selecting the $r$-largest model magnitudes, and then updating the $k$ most stale of them. By compressing each update into only \(k\) selected coordinates that fit into a single OFDM symbol, the proposed method reduces latency while keeping the most useful information. Our convergence analysis has shown that increasing the number of PS antennas improves convergence rate and accuracy. The analysis has also demonstrated a trade-off between compression quality and transmission quality: larger \(k\) reduces compression error but increases noise sensitivity, and hence an optimal value must be employed to balance both effects. Theoretical results have been validated by experimental results. Future work include analytically exploring the optimal value of \(k\) based on channel conditions and compression error, as well as extending to a fully-blind (CSI-free) system at both the PS and clients.



\begin{thebibliography}{99}

\bibitem{H. B. McMahan 2016} H. B. McMahan, E. Moore, D. Ramage, S. Hampson, and B. Aguera y Arcas. Communication-efficient learning of deep networks from decentralized data. In {\it Proc. AISTATS}, April 2017.


\bibitem{Federated learning: Strategies} J. Konečný, H. B. McMahan, F. X. Yu, P. Richtárik, A. T. Suresh, and D. Bacon. Federated learning: Strategies for improving communication efficiency. In {\it Proc. NIPS Wkshps}, 2016. 



\bibitem{privacy} K. A. Bonawitz, V. Ivanov, B. Kreuter, D. Ramage, A. Segal, A. Marcedone, H. B. McMahan, S. Patel, and K. Seth. Practical secure aggregation for privacy-preserving machine learning. In {\it Proc. ACM SIGSAC CCS}, October 2017.

\bibitem{advances and problems fl} P. Kairouz {\it et al.} Advances and Open Problems in Federated Learning. {\it Foundations \& Trends Machine Learn.}, 14(1--2):1--210, June 2021.



\bibitem{analog-digital}M. F. Ul Abrar and N. Michelusi. Analog-digital scheduling for federated learning: A communication-efficient approach. In {\it Proc. Asilomar}, October 2023.


\bibitem{delay sensitive} A. Ali and A. Arafa. Delay sensitive hierarchical federated learning with stochastic local updates. {\it IEEE Trans. Cogn. Commun. Netw.}, 11(5):3412--3424, October 2025.



\bibitem{ABS} H. H. Yang, A. Arafa, T. Q. S. Quek, and H. V. Poor. Age-based scheduling policy for federated learning in mobile edge networks. In {\it Proc. IEEE ICASSP}, May 2020.



\bibitem{value of information} M. A. Khan, H. H. Yang, Z. Chen, A. Iera, and N. Pappas. Value of information and timing-aware scheduling for federated learning. In {\it Proc. IEEE CSCN}, November 2023.

\bibitem{submodular} L. Ye and V. Gupta. Client scheduling for federated learning over wireless networks: A submodular optimization approach. In {\it Proc. IEEE CDC}, December 2021.


\bibitem{online scheduling} B. Xu, W. Xia, J. Zhang, T. Q. S. Quek, and H. Zhu. Online client scheduling for fast federated learning. {\it IEEE Wireless Commun. Lett.}, 10(7):1434--1438, July 2021.


\bibitem {two-phase deep} X. Chen et al. Toward dynamic resource allocation and client scheduling in hierarchical federated learning: A two-phase deep reinforcement learning approach. {\it IEEE Trans. Commun.}, 72(12):7798--7813, December 2024.


\bibitem{CSMAAFL} X. Ma, Q. Wang, H. Sun, R. Q. Hu, and Y. Qian. {CSMAAFL}: Client scheduling and model aggregation in asynchronous federated learning. In {\it Proc. IEEE ICC}, June 2024.


\bibitem{QSGD} D. Alistarh, D. Grubic, J. Z. Li, R. Tomioka, and M. Vojnovic. {QSGD}: Communication-efficient SGD via randomized quantization and encoding. In {\it Proc. NeurIPS}, December 2017.


\bibitem{UVeQFed} N. Shlezinger, M. Chen, Y. C. Eldar, H. V. Poor, and S. Cui. {UVeQFed}: Universal vector quantization for federated learning. {\it IEEE Trans. Signal Process.}, 69:500--514, December 2021.


\bibitem{quantization constraints} N. Shlezinger, M. Chen, Y. C. Eldar, H. V. Poor, and S. Cui. Federated learning with quantization constraints. In {\it Proc. IEEE ICASSP}, May 2020.

\bibitem{joint privacy quantization} N. Lang and N. Shlezinger. Joint privacy enhancement and quantization in federated learning. In {\it Proc. IEEE ISIT}, June 2022.

\bibitem{Hierarchical quantization} L. Liu, J. Zhang, S. Song, and K. B. Letaief. Hierarchical federated learning with quantization: Convergence analysis and system design. {\it IEEE Trans. Wireless Commun.}, 22(1):2--18, January 2023.


\bibitem{DSGD} M. M. Amiri and D. Gündüz. Machine learning at the wireless edge: Distributed stochastic gradient descent over-the-air. In {\it Proc. IEEE ISIT}, July 2019.


\bibitem {FedAQ} L. Qu, S. Song, and C.-Y. Tsui. Fed{AQ}: Communication-efficient federated edge learning via joint uplink and downlink adaptive quantization. arXiv:2406.18156.


\bibitem {OTA computation} K. Yang, T. Jiang, Y. Shi, and Z. Ding. Federated learning via over-the-air computation. {\it IEEE Trans. Wireless Commun.}, 19(3):2022--2035, March 2020.


\bibitem{Beamforming and device selection OTA} F. M. Kalarde, M. Dong, B. Liang, Y. A. E. Ahmed, and H. T. Cheng. Beamforming and device selection design in federated learning with over-the-air aggregation. {\it IEEE Open J. Commun. Soc.}, 5:1710--1723, March 2024.


\bibitem{Beamforming Vector Design} M. Kim, A. L. Swindlehurst, and D. Park. Beamforming vector design and device selection in over-the-air federated learning. {\it IEEE Trans. Wireless Commun.}, 22(11):7464--7477, November 2023.


\bibitem {Active RIS} D. Zhang, M. Xiao, M. Skoglund, and H. V. Poor. Federated learning via active {RIS} assisted over-the-air computation. In {\it Proc. IEEE ICMLCN}, May 2024.


\bibitem{Knowledge-guided OTA} Y. Zou, Z. Wang, X. Chen, H. Zhou, and Y. Zhou. Knowledge-guided learning for transceiver design in over-the-air federated learning. {\it IEEE Trans. Wireless Commun.}, 22(1):270--285, January 2023.


\bibitem{Age-based device selection} J. Liu, Z. Chang, and Y.-C. Liang. Age-based device selection and transmit power optimization in over-the-air federated learning. arXiv:2501.01828.


\bibitem{on analog} T. Sery and K. Cohen. On analog gradient descent learning over multiple access fading channels. {\it IEEE Trans. Signal Process.}, 68:2897--2911, April 2020.


\bibitem{Broadband analog} G. Zhu, Y. Wang, and K. Huang. Broadband analog aggregation for low-latency federated edge learning. {\it IEEE Trans. Wireless Commun.}, 19(1):491--506, January 2020.


\bibitem{NCAirFL} H. Wen, N. Michelusi, O. Simeone, and H. Xing. {NCAirFL: CSI}-free over-the-air federated learning based on non-coherent detection. arXiv:2411.13000.


\bibitem{Revisiting analog} H. H. Yang, Z. Chen, T. Q. S. Quek, and H. V. Poor. Revisiting analog over-the-air machine learning: The blessing and curse of interference. {\it IEEE J. Sel. Top. Signal Process.}, 16(3):406--419, April 2022.

\bibitem{Random Orthogonalization} X. Wei, C. Shen, J. Yang, and H. V. Poor. Random orthogonalization for federated learning in massive {MIMO} systems. {\it IEEE Trans. Wireless Commun.}, 23(3):2469--2485, March 2024.


\bibitem{Distributed SGD} J. Choi. Communication-efficient distributed {SGD} using random access for over-the-air computation. {\it IEEE J. Sel. Areas Inf. Theory}, 3(2):206--216, June 2022.


\bibitem{CS 1} A. Edin and Z. Chen. Over-the-air federated learning with compressed sensing: Is sparsification necessary? In {\it Proc. IEEE ICMLCN}, May 2024.


\bibitem{CS 2} Y.-S. Jeon, M. M. Amiri, and N. Lee. Communication-efficient federated learning over {MIMO} multiple access channels. {\it IEEE Trans. Commun.}, 70(10):6547--6562, October 2022.

\bibitem{rtopk} L. P. Barnes, H. A. Inan, B. Isik, and A. Özgür. r{T}op-k: A statistical estimation approach to distributed {SGD}. {\it IEEE J. Sel. Areas Inf. Theory}, 1(3):897--907, November 2020.



\bibitem{Layer} J. Oh, D. Lee, D. Won, W. Noh, and S. Cho. Communication-efficient federated learning over-the-air with sparse one-bit quantization. {\it IEEE Trans. Wireless Commun.}, 23(10):15673--15689, October 2024.


\bibitem{SegOTA} C. Zhang, M. Dong, B. Liang, A. Afana, and Y. A. Ahmed. {SegOTA}: Accelerating over-the-air federated learning with segmented transmission. arXiv:2504.09745.


\bibitem{blind} M. M. Amiri, T. M. Duman, D. Gündüz, S. R. Kulkarni, and H. V. Poor. Blind federated edge learning. {\it IEEE Trans. Wireless Commun.}, 20(8):5129--5143, August 2021.


\bibitem{age aware} R. Du, Z. Li, and H. H. Yang. Age-aware partial gradient update strategy for federated learning over the air. arXiv:2504.01357.


\bibitem{ragek} M. Mortaheb, P. Kaswan, and S. Ulukus. rAge-k: Communication-efficient federated learning using age factor. In {\it Proc. Asilomar}, October 2024.


\end{thebibliography}
\end{document}